\documentclass[12pt,final,
thmsa,
epsf,a4paper]
{article}
\usepackage{booktabs}

\newtheorem{theorem}{Theorem}

\newtheorem{definition}{Definition}
\newtheorem{example}{Example}

\newenvironment{proof}[1][Proof]{\emph{#1.} }{\  \hfill $\square $ \vspace{5 pt}}

\usepackage{natbib}
\usepackage{amssymb}
\usepackage[dvips]{graphics}
\usepackage{amsmath}

\usepackage{mathpazo}
\usepackage{enumerate}

\usepackage{amsfonts}

\usepackage{amssymb}

\usepackage{enumerate}

\usepackage{mathrsfs}

\usepackage[usenames]{color}
\usepackage{cancel}

\usepackage{pgf,tikz}

\usetikzlibrary{intersections}
\usetikzlibrary{decorations.pathreplacing,angles,quotes}
\usetikzlibrary{arrows.meta}
\usetikzlibrary{arrows, decorations.markings}
\tikzset{myptr/.style={decoration={markings,mark=at position 1 with %
       {\arrow[scale=2,>=stealth]{>}}},postaction={decorate}}}

\usepackage{hyperref}
\hypersetup{
    colorlinks,
    citecolor=blue,
    filecolor=black,
    linkcolor=blue,
    urlcolor=blue
}

\usepackage{pgf,tikz}
\usetikzlibrary{arrows}

\usepackage[utf8]{inputenc}
\usepackage{amssymb, amsmath,geometry}
\usepackage{fancyhdr}

\usepackage{soul}

\usepackage{emptypage}

\newcommand*\samethanks[1][\value{footnote}]{\footnotemark[#1]}

\usepackage[T1]{fontenc}
\DeclareFontFamily{T1}{calligra}{}
\DeclareFontShape{T1}{calligra}{m}{n}{<->s*[1.44]callig15}{}
\DeclareMathAlphabet\mathcalligra   {T1}{calligra} {m} {n}

\begin{document}

\title{Quasi-stability notions in two-sided matching models
\thanks{Thanks \ldots. We acknowledge financial support
from UNSL through grants 032016 and 030320, from Consejo Nacional
de Investigaciones Cient\'{\i}ficas y T\'{e}cnicas (CONICET) through grant
PIP 112-200801-00655, and from Agencia Nacional de Promoción Cient\'ifica y Tecnológica through grant PICT 2017-2355.}}


\author{Nadia Gui\~{n}az\'u\samethanks [2]\thanks{Instituto de Matem\'{a}tica Aplicada San Luis (UNSL-CONICET) and Departamento de Matemática, Universidad Nacional de San
Luis, San Luis, Argentina. Emails: \texttt{ncguinazu@unsl.edu.ar} (N. Gui\~{n}azu), \texttt{nmjuarez@unsl.edu.ar} (N. Juarez), \texttt{paneme@unsl.edu.ar} (P. Neme) and \texttt{joviedo12@gmail.com} (J. Oviedo).}    \and Noelia Juarez\samethanks[2] \and Pablo Neme\samethanks[2] \and Jorge Oviedo}

\date{\today}
\maketitle

\begin{abstract}
This paper presents weakened notions of corewise stability and setwise stability for matching markets where agents have substitutable choice functions. We introduce the concepts of worker-quasi-core, firm-quasi-core, and worker-quasi-setwise stability. We also examine their relationship to established notions in the literature, such as worker-quasi and firm-quasi stability in both many-to-one and many-to-many markets.

\bigskip

\bigskip


\end{abstract}

\section{Introduction}

In this paper, we introduce new stability concepts in two-sided matching markets that extend existing quasi-stability notions to encompass corewise and setwise stability. Specifically, we present the notions of worker-quasi-core, firm-quasi-core, worker-quasi-setwise stable, and firm-quasi-setwise stable, which provide a new framework for understanding stability in both many-to-one and many-to-many matching markets. By broadening quasi-stability beyond pairwise stability, our approach captures more flexible stability requirements suitable for practical applications where agents may have multiple, substitutable partners. This contribution opens new perspectives on stability theory, linking traditional core-based solutions with quasi-stable structures.

Labor markets, a central application in matching theory, are typically analyzed through two main frameworks: the many-to-one model, where workers are matched to a single firm, and the many-to-many model, where workers may have multiple jobs simultaneously. This paper examines both models assuming that each firm operates with a choice function satisfying the substitutability property. Substitutability is fundamental in ensuring pairwise stability in markets with multiple agents, as it guarantees that a firm’s choice of a worker remains consistent even if other workers become unavailable \citep[see][for more details]{kelso1982job}. Workers are similarly assumed to have substitutable choice functions, which, in the many-to-one model, correspond to strict preferences over single-firm sets and the empty set.\footnote{This is usually found in the literature as a preference relation over individual firms. However, since in this paper, the preference relation is related to a choice function, we consider the preference relation to be over single-firm sets.}

In matching markets, stability is often framed in terms of cooperative game theory, where agents form mutually beneficial arrangements based on shared incentives. Traditionally, three central stability concepts have structured the analysis of these markets: pairwise stability, corewise stability, and setwise stability. Pairwise stability, one of the most widely used concepts, ensures that no firm-worker pair, unlinked in the current matching, can improve their outcomes by matching with each other. Building on this, \citet{blair1988lattice} developed an order among matchings that characterizes pairwise stable matchings as having a lattice structure. Blair’s order captures agents' preferences in a hierarchical form: a set of agents is \emph{(Blair)-preferred} to another if the choice function selects it when both sets are available.

Corewise and setwise stability are broader concepts that incorporate coalition-based improvements within a matching. A matching is in the \emph{core} if no coalition can coordinate to form a new arrangement in which each agent within the coalition is matched with other coalition members, each agent's new set of partners is at least as preferred as their current match, and at least one agent strictly prefers their new set. A matching \emph{setwise dominates} another if a coalition of agents---formed only by agents obtaining new partners---can arrange an alternative matching where each agent's new set is weakly preferred to their current match, and at least one agent strictly prefers it. Thus, the set of matching not setwise dominated is called the \emph{setwise stable set}. In many-to-one models, corewise and setwise stability often align with pairwise stability, leading to unified stability results. However, in many-to-many models, these stability notions diverge and introduce distinct layers of complexity, as the involvement of multiple firms and workers creates nuanced coalition interactions \citep[see][for more details]{sotomayor1999three,echenique2004theory}.

Motivated by recent work in quasi-stability, which offers a flexible approach by allowing limited deviations in blocking pairs, we extend this flexibility to corewise and setwise stability. These quasi-stability concepts have gained attention because they relax the traditional pairwise stability notion by allowing blocking pairs under certain constraints, such as preventing firms from dismissing existing employees. This relaxation allows us to address cases where pairwise stability might be unattainable or unrealistic due to market constraints or agents’ preferences. For instance, \citet{blum1997vacancy}, \citet{wu2018lattice}, \cite{bonifacio2022lattice}, \cite{bonifacio2024envyfreelattice}, \cite{bonifacio2024lattice}, and \cite{yang2024existence}  explore quasi-stability notions in various market settings, showing how quasi-stable matchings can retain desirable properties, such as lattice structure.

Building on this foundation, we extend quasi-stability concepts to define the worker-quasi-core, firm-quasi-core, worker-quasi-setwise stable, and firm-quasi-setwise stable matchings. A matching belongs to the \emph{worker-quasi-core} if, in cases where a coalition dominates the current matching, workers in the coalition retain their current firms and potentially add new ones in the dominant matching. In many-to-one models, this definition simplifies to require that only those workers in a deviating coalition who switch partners from the dominated to the dominant matching are unemployed in the dominated matching. Analogously, we define the \emph{firm-quasi-core}, which imposes that firms in a dominating coalition retain their current employees, adding only new hires in the dominant matching. These definitions, which parallel the firm- and worker-based quasi-stability in pairwise stability, adapt corewise stability to a framework that accommodates more flexible market conditions.

We extend these ideas to setwise stability by introducing the worker-quasi-setwise stable and firm-quasi-setwise stable matchings. A matching is \emph{worker-quasi-setwise stable} if it is individually rational and, when another matching setwise dominates it through a coalition, the workers in the coalition retain their existing firms while incorporating new ones in the dominant matching. Both concepts are particularly relevant for many-to-many models, where setwise and quasi-stability provide meaningful insights into complex market interactions. Due to the symmetry of the market, the \emph{firm-quasi-setwise stable} notion is defined analogously.

Our results provide new characterizations within the quasi-stable framework. Specifically, we establish that in many-to-one models, individually rational matchings in the worker-quasi-core are precisely worker-quasi-stable matchings. We also show that individually rational matchings in the firm-quasi-core are firm-quasi-stable matchings. However, as is often the case in the matching literature, these characterizations do not extend directly to the many-to-many model due to the increased complexity of multiple employment relationships. Nevertheless, we establish that worker-quasi-stable matchings remain within the worker-quasi-core, and this inclusion can be strict depending on the market's characteristics. By symmetry, this result also holds for firm-quasi-stable matchings within the firm-quasi-core.

Finally, we characterize the (firm) worker-quasi-setwise stable set using (firm) worker-quasi-stable matchings. Moreover, we show that a matching that is both worker-quasi-stable and firm-quasi-stable is either pairwise stable or is dominated by another matching via some coalition, highlighting the nuanced relationships between different quasi-stability notions and traditional stability concepts. 

The paper is structured as follows. Section \ref{seccion de preliminares} provides preliminary concepts. Section \ref{sección many-to-one } introduces two weakened forms of corewise stability for many-to-one models: the worker-quasi-core and firm-quasi-core, each with characterizations in terms of quasi-stable matchings. Section \ref{sección many-to-many} extends these notions to many-to-many models, showing that individually rational matchings in the worker-quasi-core do not fully characterize the worker-quasi-stable matchings, i.e. only one inclusion holds. We also introduce the worker-quasi-setwise stable set and its characterization. Finally, Section \ref{Concluding Remarks} offers concluding remarks.

\section{Preliminaries}\label{seccion de preliminares}In this paper we consider both a many-to-many and many-to-one matching markets where there are two disjoint sets of agents: the set of firms $F$ and the set of workers $W$. Each firm $f\in F$ has a choice function $C_f$ over the set of all subsets of $W$ that satisfies \textbf{substitutability}, i.e.,   for $T'\subseteq T \subseteq W$, we have $C_f(T)\cap T' \subseteq C_f(T')$.\footnote{Substitutability is equivalent to the following: for each $w\in W$ and each $T\subseteq W$ such that $w\in T$, $w\in C_f\left(T\right)$ implies that $w\in C_f\left(T'\cup \{w\}\right)$ for each $T' \subseteq T.$} In addition, we assume that $C_f$ satisfies $C_f(T')=C_f(T)$ whenever $C_f(T)\subseteq T' \subseteq T \subseteq W.$ This property is known in the literature as \textbf{consistency}. If $C_f$  satisfies substitutability and consistency, then it also satisfies \begin{equation}\label{propiedad de choice}
C_f\left(T\cup T'\right)=C_f\left(C_f\left(T\right)\cup T'\right)
\end{equation} for each pair of subsets $T$ and $T'$ of $W$.\footnote{This property is known in the literature as \emph{path independence} \citep[see][]{alkan2002class}.}


For the many-to-many case, each worker $w\in W$ has a choice function $C_w$ over the set of all subsets of $F$ that satisfies substitutability, and consistency (and, therefore, path independence). On the other hand, in the many-to-one case, each worker \( w \in W \) has a strict preference relation \( >_w \) over single-firm sets and the empty set.
 Since the sets $F$ and $W$ are kept fixed throughout the paper, we denote a many-to-many matching market by  $(C_F,C_W)$, and a many-to-one matching market by $(C_F,>_W),$ where $>_W$ is the preference profile for workers in the many-to-one market, and $C_F$ and $C_W$ are the profiles of choice functions for all firms and all workers, respectively. Given a firm $f$, it is well established that a choice function $C_f$ may induce a preference relation $>_f$, which is not necessarily unique \citep[see][for more details]{martinez2008invariance,martinez2012invariance}. In contrast, a preference relation $>_f$ induces only one choice function $C_f$ \citep[see][for more details]{alkan2002class,echenique2004theory}. 

In the remainder of this section, we provide several standard definitions for a many-to-many matching market. These definitions can be readily adapted to the many-to-one case with minor modifications.

\begin{definition}
A \textbf{matching} $\mu$ is a function from the set $F\cup W$ into $2^{F\cup W}$ such that for each $w\in W$ and for each $f\in F$:
\begin{enumerate}[(i)]
\item $\mu(w)\subseteq F$,
\item $\mu(f)\subseteq W$,
\item $w\in \mu(f)$ if and only if $f\in \mu(w)$.
\end{enumerate}
\end{definition}

Let $\boldsymbol{\widetilde{\mathcal{M}}}$ and $\boldsymbol{\mathcal{M}}$ denote the set of matchings for markets  $(C_F,C_W)$ and $(C_F,>_W)$, respectively.

Note that for a many-to-one matching, we need to require in (i) that  $\mu(w)\subseteq F$ and $|\mu(w)|\leq 1$.
Agent $a\in F\cup W$ is \textbf{matched} if $\mu(a) \neq \emptyset$, otherwise she is \textbf{unmatched}. 
A matching $\mu$ is \textbf{blocked by agent $\boldsymbol{a}$} if $\mu(a)\neq C_a(\mu(a))$.\footnote{In the many-to-one case, this condition for a worker $w$ is equivalent to stating that $\emptyset >_w \mu(w).$} A matching is  \textbf{individually rational} if it is not blocked by any individual agent. Let $\boldsymbol{\widetilde{\mathcal{I}}}$ and $\boldsymbol{\mathcal{I}}$ denote the set of individually rational matchings for markets $(C_F,C_W)$ and $(C_F,>_W)$, respectively.  A matching $\mu$ is \textbf{blocked by a firm-worker pair $\boldsymbol{(f,w)}$} if $w \notin \mu( f ), w \in C_f(\mu( f )\cup \{w\}),$ and $f \in C_w(\mu( w )\cup \{f\})$.\footnote{In the many-to-one case, this condition for a worker $w$ is equivalent to stating that $\{f\} >_w \mu(w).$ } A matching $\mu$ is \textbf{pairwise stable} if it is not blocked by any individual agent or any firm-worker pair. Let $\boldsymbol{\widetilde{\mathcal{S}}}$ and $\boldsymbol{\mathcal{S}}$ denote the set of stable matchings for markets $(C_F,C_W)$ and $(C_F,>_W)$, respectively. 

 Given a firm $f$, \cite{blair1988lattice} defines a partial order for $f$  over  subsets of workers  as follows: given  firm $f$'s choice function $C_f$ and two subsets of workers  $T$ and  $T'$, we write $\boldsymbol{T \succeq_f T'}$ whenever $T=C_f(T \cup T')$, and $\boldsymbol{T \succ_f T'}$ whenever $T \succeq_f T'$ and $T \neq T'$. Given a worker $w$ and $C_w$, we analogously define $\succeq_w$ and $\succ_w$. Given Blair's partial order and two  matchings $\mu$ and  $\mu'$,  we write  $\boldsymbol{\mu \succeq_F \mu'}$ whenever $\mu(f) \succeq_f \mu'(f)$ for each $f\in F$, and we write $\boldsymbol{\mu \succ_F \mu'}$ if, in addition, $\mu \neq \mu'$.\footnote{We call $\succeq_F$ the unanimous Blair order for the firms.}  Similarly, we define $\succeq_W$ and $\succ_W$. In the case of a many-to-one market, we denote $\geq_W$ as the unanimous order of all workers. \footnote{Note that in this case, the unanimous Blair order $\succeq_W$ coincides with the unanimous order $\geq_W$.}

In addition to pairwise stability, another solution concept considered in this paper is the ``core''. Before formally defining the core, we introduce the notion of dominance among matchings through a coalition of agents. Given a subset of agents $S$, let $\mu(S)$ be the set of partners of each agent in $S$, i.e., $\mu(S)=\{\mu(a): a\in S\}.$
\begin{definition}\label{defino dominancia}
    Let $\mu,\mu'$ be two different matchings, and let $S\subseteq F\cup W$ be a non-empty coalition. We say that  \textbf{$\boldsymbol{\mu'$ dominates $\mu}$ via }  $\boldsymbol{S}$ if  $\mu'(S)\subseteq S$,  $\mu'(a)\succeq_a\mu(a)$  for each $a\in S$, and there is $a'\in S$ such that $\mu'(a')\succ_{a'}\mu(a')$.

\end{definition}

We define the  \textbf{core} as the set of undominated matchings.\footnote{The notion of corewise stability that we consider in this paper is referred to as ``Blair corewise stability'', as used, for instance, in \cite{echenique2004theory}.} Note that for the many-to-one market, the condition  $\mu'(w)\succeq_w\mu(w)$ is equivalent to $\mu'(w)\geq_w\mu(w)$ for $w\in S$.  Thus, we denote the core by $\boldsymbol{\widetilde{\mathcal{C}}}$ for market $(C_F,C_W)$ and  by $\boldsymbol{\mathcal{C}}$ for market $(C_F,>_W)$. 
 Furthermore, in many-to-one markets, the core coincides with the pairwise stable set and, therefore, is non-empty.

\section{Quasi-stability notions in a many-to-one model}\label{sección many-to-one }
In this section, we present two weakening of the corewise stability notion for a many-to-one market. In Subsection \ref{subseccion worker quasi en M-1}, we introduce the notion of ``worker-quasi-core'' and establish its relationship with other well-known solution concepts in the literature. In Subsection \ref{subseccion firma quasi en M-1}, we weaken the corewise stability notion in the direction of firms by presenting the ``firm-quasi-corewise stability'' notion and examine its relationships with other established solution concepts.

\subsection{Worker-quasi-corewise stability}\label{subseccion worker quasi en M-1}

In this subsection, we begin by presenting a weakening of the corewise stability notion, which allows some domination while imposing restrictions on some workers involved in the formed coalition.

\begin{definition}\label{defino worker-quasi-core}
    A matching $\mu$ is in the \textbf{worker-quasi-core} if there are a matching $\mu'$ and a coalition $S$ such that $\mu'$ dominates $\mu$ via $S$, then  $\mu(w)=\emptyset$ for each $w\in S $ that satisfies  $\mu(w)\neq\mu'(w).$
\end{definition}
Let $\boldsymbol{ \mathcal{C^{QW}} }$ denote the worker-quasi-core for market $(C_F,>_w)$. A matching is in the worker-quasi-core if it is either undominated or in the case that is dominated by another matching via some coalition, the workers in that coalition that deviates to improve their situation in the new matching must be unemployed in the dominated matching. 

Notice that, by Definition \ref{defino worker-quasi-core}, the worker-quasi-corewise stability notion is a weakening of the corewise stability notion, i.e. the worker-quasi-core contains the core. As mentioned in the previous section, the core is non-empty, and thus $\mathcal{{C^{QW}}}$ is also non-empty. Moreover, in the many-to-one model, the core is always individually rational \citep[see][for more details]{echenique2004theory}. However, the following example shows that the worker-quasi-core may not be individually rational even in a one-to-one model.
\begin{example}
     Let $(C_F,>_w)$ be a market where   $F = \{f\}$ and $W = \{w\}$. The choice function for $f$  fulfills that $C_f(\{w\})=\emptyset$. As previously mentioned, this choice function induces a preference for firm $f$. Thus, consider the following preference profile: $$>_{f}:  \emptyset,\{w\},  \text{~~~ and~~~ } >_{w}: \{f\},\emptyset .$$ 
  Consider the following matching:
 $$
\mu=\begin{pmatrix}
f  \\
w\\
\end{pmatrix}
.~$$
It is easy to see that $\mu\notin \mathcal{I}$. Note that the unique $\mu'$ that dominates $\mu$ via a coalition $S$ is such that $\mu'(f)=\emptyset$ with $S=\{f\}$. Since there are no workers in $S$, the condition that $\mu(w)=\emptyset$ for each $w\in S$  such that $\mu(w)\neq\mu'(w)$ is trivially satisfied.    Therefore,
$\mu\in \mathcal{C^{QW}} $.  \hfill $\Diamond$
\end{example}

We now focus on studying the relationship between worker-quasi-core matchings and worker-quasi-stable matchings. Worker-quasi-stability is a weakening of pairwise stability that allows blocking pairs involving a firm and an unemployed worker. Formally,

\begin{definition}\label{def worker-quasi many-to-one}
    A matching $\mu$ is \textbf{worker-quasi-stable} if it is individually rational and, whenever $(f,w)$ blocks $\mu$, we have $\mu(w)=\emptyset$. 
\end{definition}

Let $\boldsymbol{\mathcal{QW}}$ denote the set of all worker-quasi-stable matchings for market $(C_F,>_{W}).$  Notice that the set $\mathcal{QW}$ is always non-empty since the empty matching in which each agent is unmatched belongs to this set.



Let us observe that worker-quasi-stable matchings, although a weakening of stable matchings remain individually rational.\footnote{Recall that the set of worker-quasi-stable matchings contains the pairwise stable set.} Therefore, to study the relationship between matchings in the worker-quasi-core and worker-quasi-stable matchings, we must restrict our analysis to those matchings in the worker-quasi-core that satisfy individual rationality. Thus, for a many-to-one setting, we present a characterization of the set of individually rational matchings in the worker-quasi-core through the set of worker-quasi-stable matchings.

\begin{theorem}\label{teorema characterization de quasi stable en m-1}
 Let $(C_F,>_w)$ be a many-to-one market, then     $\mathcal{I}\cap \mathcal{C^{QW}}=\mathcal{QW}.$
\end{theorem}
\begin{proof} We will prove the double inclusion in two cases.
\begin{itemize} 
    \item[$\boldsymbol{\subseteq )}$] 
     Let $\mu\in \mathcal{I}\cap\mathcal{C^{QW}}$. If $\mu\in \mathcal{S}$, $\mu\in \mathcal{QW}$ and we are done. Otherwise, assume that $(f,w')$ is a blocking pair of $\mu$. We consider $S=\{f,C_f(\mu(f)\cup\{w'\})\}$ and let $\mu'$ be a matching such that $\mu'(f)=C_f(\mu(f)\cup \{w'\})$. We claim that $\mu'$ dominates $\mu$ via $S$. To see this,  we first show that $\mu'(S)\subseteq S$. Let $w\in S$. Thus, $w \in C_f(\mu(f)\cup \{w'\})$. By definition of $\mu'$, we have that $w\in\mu'(f).$ Since $\mu'$ is a matching, this implies that $\mu'(w)=\{f\}$. Since $f\in S$, we have that $\mu'(w)\in S$. Take $f\in S.$ By definition of $\mu'$,  $\mu'(f)=C_f(\mu(f)\cup\{w'\})$. Hence,  $\mu'(S)\subseteq S$. Now, we prove that $\mu'(w)\geq_w \mu(w)$ for each $w\in S$. Since $(f,w')$ is a blocking pair of $\mu$, we have that $w'\in C_f(\mu(f)\cup\{w'\})$ and $\{f\}>_{w'}\mu(w')$, and by definition of $\mu'$, $\mu'(w')=\{f\}$.  On the other hand, for $w\neq w'$ we have that $\mu'(w)=\mu(w)=\{f\}$. So, $\mu'(w)\geq_w \mu(w)$ for each $w\in S$. Lastly, we prove $\mu'(f)\succeq_f \mu(f)$ for $f\in S$. To see this,  by Path Independence and definition of $\mu'$, we have  $C_f(\mu(f)\cup\mu'(f))=C_f(\mu(f)\cup C_f(\mu(f)\cup \{w'\}))=C_f(\mu(f)\cup\{w'\})=\mu'(f).$ This implies that $\mu'(f)\succeq_f \mu(f)$. Hence, $\mu'$ dominates $\mu$ via $S$ proving the claim. Furthermore, since  $\mu\in \mathcal{I}\cap\mathcal{C^{QW}}$ and $\mu'(w')\neq \mu(w')$, we have that $\mu(w')=\emptyset.$ Therefore $\mu\in \mathcal{QW}.$
  \item [$\boldsymbol{\supseteq )}$] 
  Assume that $\mu\in \mathcal{I}\setminus \mathcal{C^{QW}}$ and we will prove that $\mu \notin \mathcal{QW}.$  Then, there are a matching $\mu'$, a coalition $S$, and a worker $w\in S$ satisfying that $\mu'$ dominates $\mu$ via $S$, $\mu'(w)\neq\mu(w)$, but $\mu(w)\neq\emptyset$. Since $\mu\in \mathcal{I}$  and $\mu'$ dominates $\mu$ via $S$, there is $f\in S$ such that $\{f\}=\mu'(w)$. Given that $\mu'(w)\neq \mu(w)$, we have that \begin{equation}\label{ecu 1 teorema 1}
       \{f\}=\mu'(w)>_w\mu(w).
    \end{equation} Since  $\mu'(f)\succeq_f\mu(f)$, $w\in C_f(\mu(f)\cup\mu'(f))$ and, the fact that $w\in \mu'(f)\setminus \mu(f)$ implies, by substitutability, $w\in C_f(\mu(f)\cup\{w\}).$ This last fact together with \eqref{ecu 1 teorema 1} imply that $(f,w)$ is a blocking pair of $\mu$ with  $\mu(w)\neq\emptyset.$ Therefore $\mu \notin \mathcal{QW}. $
\end{itemize}
By the double inclusion we have $\mathcal{I}\cap \mathcal{C^{QW}}=\mathcal{QW}.$
\end{proof}

The relation among the individuality rational matching set, the worker-quasi-core set, the worker-quasi-stable matching set, and the pairwise stable matching set is depicted in Figure \ref{Diagrama many to one worker quasi}.

\begin{figure}[ht!]
    \centering
\begin{tikzpicture}[scale=0.8]
    \fill[red, opacity=0.2] (2, 0.5) circle (4.5cm);
    \draw (2, 0.5) circle (4.5cm);
    \node at (4.8, 1) {$\mathcal{I}$};

     \fill[red, opacity=0.2] (-1, 0) circle (5cm);
    \draw (-1, 0) circle (5cm);
    \node at (-4, 0) {$\mathcal{C^{QW}}$};
    \node at (0.8, 3) {$\mathcal{QW}$};

   \fill[red, opacity=0.2] (1, 0) circle (2cm);
    \draw (1, 0) circle (2cm);
    \node at (1,0) {$\mathcal{C} = \mathcal{S} $};
\end{tikzpicture}
 \caption{Relation among worker-quasi stability notions in many-to-one markets}
    \label{Diagrama many to one worker quasi}
\end{figure}
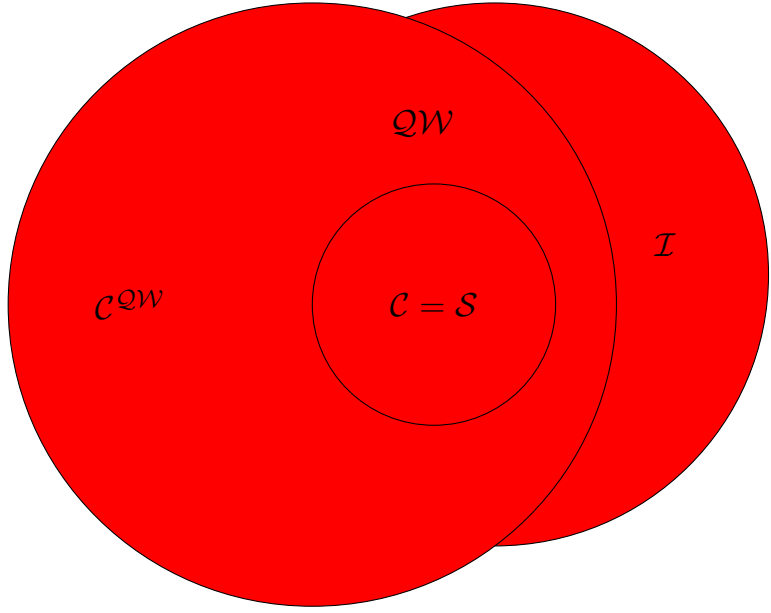


\subsection{Firm-quasi-corewise stability}\label{subseccion firma quasi en M-1}

In this subsection, we present another weakening of the corewise stability notion, in which we
allow some domination while imposing restrictions on the firms involved in the coalition that is formed.

\begin{definition}
    A matching $\mu$ is in the \textbf{firm-quasi-core} if there are a matching $\mu'$ and a coalition $S$ such that $\mu'$ dominates $\mu$ via $S$, it must satisfy that $\mu(f)\subseteq\mu'(f)$ for each $f\in S$. 
\end{definition}

Let $\boldsymbol{ \mathcal{C^{QF}} }$ denote the firm-quasi-core for market $(C_F,>_w)$. A matching is in the firm-quasi-core if it is either undominated or in the case that is dominated by another matching via some coalition, the firms in that coalition that deviate to improve their situation in the new matching only incorporate new workers without firing anyone. 
Notice that, by definition, the firm-quasi-core contains the core.  Since $\mathcal{{C}}\subseteq \mathcal{C^{QF}}$ and the fact that in many-to-one settings $\mathcal{C}\neq \emptyset$, we have that  $\mathcal{{C^{QF}}}\neq \emptyset$. Recall that, in the many-to-one model, the core is always individually rational, however, the following example shows that the firm-quasi-core may not be individually rational even in a one-to-one market.

\begin{example}
     Let $(C_F,>_w)$ be a market where   $F = \{f\}$ and $W = \{w\}$. The choice function for $f$ satisfy $C_f(\{w\})=\{w\}.$ As previously mentioned, these choice functions induce a preference relation for firm $f$, and a preference for worker $w$.  Thus, we obtain the following preference profile: $$>_{f}: \{w\}, \emptyset  \text{~~~ and~~~ } >_{w}: \emptyset,\{f\} .$$ 
   Consider the following matching:
 $$
\mu=\begin{pmatrix}
f  \\
w\\
\end{pmatrix}
.$$
It is easy to see that $\mu\notin \mathcal{I}$. Note that the unique $\mu'$ that dominates $\mu$ via a coalition $S$ is such that $\mu'(w)=\emptyset$ with $S=\{w\}$. Since there are no firms in $S$, the condition that $\mu(f)\subseteq \mu'(f)$ for each $f\in S$ is trivially satisfied.    Therefore,
$\mu\in \mathcal{C^{QF}} $.  \hfill $\Diamond$
\end{example}

To introduce the second characterization of this section, we need to define the concept of firm-quasi-stable matchings. The notion of firm-quasi-stability was first introduced for many-to-one markets in \cite{cantala2004restabilizing}.

Given $\mu \in \mathcal{M}$ and $w \in W$, let $$W_f^\mu=\{w \in W : \{f\} >_w \mu(w) \}.$$

\begin{definition}\label{def firm-quasi many-to-one}
   A matching $\mu$ is \textbf{firm-quasi-stable} if it is individually rational and, for each $f \in F$ and each $T \subseteq W_f^\mu$, we have $$\mu(f) \subseteq C_f(\mu(f)\cup T).$$ 
\end{definition}

Note that, in a firm-quasi-stable matching no firm $f$ will fire any worker assigned under $\mu$ in order to hire some of the workers who wish to be employed by firm $f$.
Denote by $\boldsymbol{\mathcal{QF}}$ the set of all firm-quasi-stable matchings for market $(C_F,>_w)$. Notice that $\mathcal{QF}$ is non-empty, since the empty matching, in which every agent is unmatched, belongs to $\mathcal{QF}$.\footnote{Recall that the set of firm-quasi-stable matchings contains the pairwise stable set.}

    




    

Next, we present our second characterization result of the many-to-one setting. We characterize of the set of individually rational matchings in the firm-quasi-core through the set of the firm-quasi-stable matchings.

\begin{theorem}
 Let $(C_F,>_w)$ be a many-to-one market, then    $\mathcal{I}\cap \mathcal{C}^{QF}=\mathcal{QF}.$
\end{theorem}
\begin{proof} We will prove the double inclusion in two cases.
\begin{itemize} 
    \item[$\boldsymbol{\subseteq )}$]  Let $\mu\in \mathcal{I}\cap \mathcal{C}^{QF}$.  If $\mu\in \mathcal{S}$,  $\mu\in \mathcal{QF}$ and we are done. Otherwise, there is a blocking pair $(f,w)$ of $\mu$. Thus, there is $T\neq \emptyset$ such that $T\subseteq W^\mu_f.$ We consider $S=\{f,C_f(\mu(f)\cup T)\}$ and let $\mu'$ be a matching such that 
$\mu'(f)=C_f(\mu(f)\cup T).$ We claim that $\mu'$ dominates $\mu$ via $S.$ To see this, we first show that $\mu'(S)\subseteq S.$ Note that $\mu'(f)=C_f(\mu(f)\cup T)$ implies that  $\mu'(f)\subseteq S.$ Let $w\in S$. Thus, $w\in C_f(\mu(f)\cup T)=\mu'(f).$ This implies that $\mu'(w)=\{f\}\subseteq S$ and, thus $\mu'(S)\subseteq S.$

Second, we show that $\mu'(f)\succeq_f \mu(f)$ for $f\in S$. By definition of $\mu'$, $C_f(\mu(f)\cup\mu'(f))=C_f(\mu(f)\cup C_f(\mu(f)\cup T))$. By Path Independence   $C_f(\mu(f)\cup C_f(\mu(f)\cup T))= C_f(\mu(f)\cup\mu(f)\cup T)=C_f(\mu(f)\cup T).$ By definition of $\mu'$ again,  $C_f(\mu(f)\cup T)=\mu'(f).$ Hence,  $C_f(\mu(f)\cup\mu'(f))=\mu'(f)$ implying that $\mu'(f)\succeq_f \mu(f)$. 

Finally, we show that $\mu'(w)\geq_w \mu(w)$ for each $w\in S.$ Let $w\in S$. Thus, $w\in C_f(\mu(f)\cup T).$ Then,  $w\in\mu(f)$ or $w\in T.$ If $w\in \mu(f)$, $\{f\}=\mu(w)=\mu'(w).$ If $w\in T$,  since $T\subseteq W^\mu_f$ then $\mu'(w)=\{f\}>_w\mu(w).$
Therefore, $\mu'$ dominates $\mu$ via $S$, proving the claim.

Since $\mu'$ dominates $\mu$ via $S$ and $\mu\in \mathcal{C^{QF}}$,  $\mu(f)\subseteq \mu'(f).$ By definition of $\mu'$, 
$\mu(f)\subseteq C_f(\mu(f) \cup T)$. This last fact together with the individual rationality of $\mu$ imply that $\mu\in \mathcal{QF}.$ Therefore,   $ \mathcal{I}\cap \mathcal{C}^{QF}\subseteq\mathcal{QF}$.

  \item[$\boldsymbol{\supseteq)}$] Assume that $\mu\in \mathcal{I} \setminus  \mathcal{C}^\mathcal{QF}$  and we will prove that $\mu \notin \mathcal{QF}.$ Then, there are a matching $\mu'$, a coalition $S$, and a firm $f\in S$ such that $\mu'$ dominates $\mu$ via $S$, and $\mu(f)\nsubseteq \mu'(f)$. We claim that $\mu'(f)\neq\emptyset.$ Assume, otherwise, i.e. $\mu'(f)= \emptyset.$ Since $\mu'$ dominates $\mu$ via $S$ and the individual rationality of $\mu$, $\emptyset=\mu'(f)=C_f(\mu(f)\cup \mu'(f))=C_f(\mu(f))=\mu(f).$ Then, $\mu(f)=\emptyset$ contradicting that $\mu(f)\nsubseteq \mu'(f).$ Therefore, $\mu'(f)\neq\emptyset,$ proving the claim. 
  
  Now, we claim that $\mu'(f)\setminus \mu(f) \neq \emptyset$. Otherwise, since $\mu'(f)\neq\emptyset$, we have that  $\mu'(f)\subseteq \mu(f).$ Moreover, by the individual rationality of $\mu$, we have that $C_f(\mu(f)\cup\mu'(f))=C_f(\mu(f))=\mu(f),$ contradicting that $\mu'(f) \succeq_f \mu(f)$ for $f\in S$. Thus, $\mu'(f)\setminus \mu(f) \neq \emptyset$, proving the claim.
  \medskip
  
  \noindent \textbf{Claim: $\boldsymbol{\mu'(f)\setminus \mu(f) \subseteq W_f^{\mu}}.$} To see this, let $w'\in \mu'(f)\setminus \mu(f).$  
  Since $\mu'(S)\subseteq S$ and $f\in S$, $w' \in S.$ Then, $\{f\}=\mu'(w')\geq_{w'}\mu(w').$ Since  $w'\in \mu'(f)\setminus \mu(f)$, we have that  $\{f\} >_{w'}\mu(w')$, that in turns implies that $w'\in W_f^{\mu}.$ Hence, $\mu'(f)\setminus \mu(f) \subseteq W_f^{\mu}$ proving  the claim. \medskip
  
  Since $\mu(f)\nsubseteq \mu'(f)$, there is $w\in \mu(f)\setminus\mu'(f).$ Note that $w\notin \mu'(f)$ and $\mu'(f) \succeq_f \mu(f)$, we have that $w\notin \mu'(f)=C_{f}(\mu(f)\cup \mu'(f)).$ By subtitutability and the fact that $\mu'(f)\setminus \mu(f) \subseteq W_f^{\mu}$, $w\notin C_{f}(\mu(f)\cup W_f^\mu).$ This last fact together with $w\in \mu(f)\setminus\mu'(f)$ imply that $\mu(f)\nsubseteq C_{f}(\mu(f)\cup W_f^\mu).$  Then, $\mu \notin \mathcal{QF}.$ Therefore,   $\mathcal{QF}\subseteq \mathcal{I}\cap \mathcal{C}^{QF}$.
  \end{itemize} 
  
 By the double inclusion we have $\mathcal{I}\cap \mathcal{C}^{QF}=\mathcal{QF}$.
\end{proof}

The relation among the individuality rational matching set, the firm-quasi-core set, the firm-quasi-stable matching set, and the pairwise stable matching set is depicted in Figure \ref{Diagrama many to one firm quasi}. Note that, the inclusions are similar to the worker-quasi-stability case.

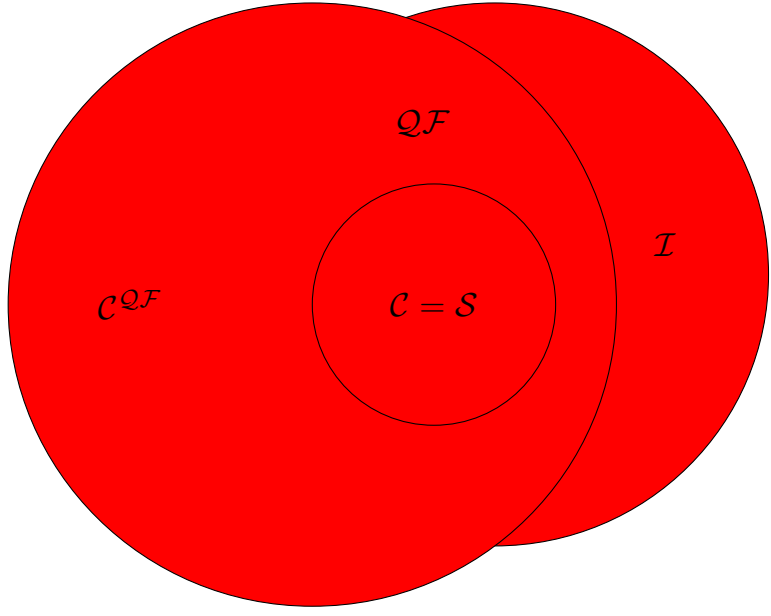
\begin{figure}[ht!]
    \centering
    \begin{tikzpicture}[scale=0.8]
    \fill[red, opacity=0.2] (2, 0.5) circle (4.5cm);
    \draw (2, 0.5) circle (4.5cm);
    \node at (4.8, 1) {$\mathcal{I}$};

     \fill[red, opacity=0.2] (-1, 0) circle (5cm);
    \draw (-1, 0) circle (5cm);
    \node at (-4, 0) {$\mathcal{C^{QF}}$};
    \node at (0.8, 3) {$\mathcal{QF}$};

   \fill[red, opacity=0.2] (1, 0) circle (2cm);
    \draw (1, 0) circle (2cm);
    \node at (1,0) {$\mathcal{C} = \mathcal{S} $};
\end{tikzpicture}
    \caption{Relation among firm-quasi stability notions in many-to-one markets}
    \label{Diagrama many to one firm quasi}
\end{figure}

\section{Quasi-stability notions in a many-to-many model}\label{sección many-to-many}

In this section, we study the relationships between different solution concepts in many-to-many markets. Specifically, we examine the connections between the quasi-core, the set of worker-quasi-stability, the setwise-stable set, the worker-quasi-stable and firm-quasi-stable sets, the pairwise stable set, and the core.

This section is divided into three subsections. In Subsection \ref{subseccion de quasi-core m-m}, we show the relationship between worker-quasi-stable matchings and individually rational matchings that belong to the worker-quasi-core. In Subsection \ref{subseccion de setwise en m-m}, we establish the connection between worker-quasi-stable matchings and setwise-stable matchings. Finally, in Subsection \ref{subseccion de core y quasi estabilidad en m-m}, we analyze the relation between worker-quasi-stable set, firm-quasi-stable set, pairwise stable matching set, and the core.

Given that the notion of worker-quasi-stability is used throughout the three subsections, we present it below.

Taking into account the notion of firm-quasi-stability introduced by \cite{cantala2004restabilizing}, which we also formally present in Definition \ref{def firm-quasi many-to-one}, the concept of worker-quasi-stability is easily adapted to the many-to-many setting. Blocking pairs are permitted as long as they do not compromise the existing relations of firms in the matching.

For $\mu \in \widetilde{\mathcal{M}}$ and $f \in F$, let $$\widetilde{F}_w^\mu=\{f \in F : w \in C_f(\mu(f)\cup\{w\})\}.$$

\begin{definition}\label{def worker-quasi many-to-many}
   A matching $\mu$ is \textbf{worker-quasi-stable} if it is individually rational, and for each $w \in W$ and each $K \subseteq \widetilde{F}_w^\mu$ we have $$\mu(w) \subseteq C_w(\mu(w)\cup K).$$ 
\end{definition}
Note that, in a worker-quasi-stable matching \( \mu \), if a worker is employed by another firm that wishes to hire her, she will not resign from any job under \( \mu \).
Denote by $\boldsymbol{\widetilde{\mathcal{QW}}}$ the set of all worker-quasi-stable matchings for market $(C_F,C_W)$. Since the empty matching belongs to this set, $\widetilde{\mathcal{QW}}\neq \emptyset$.
 Given the symmetry of the market, the notion of firm-quasi-stability is straightforward. So, denote by $\boldsymbol{\widetilde{\mathcal{QF}}}$ the set of all firm-quasi-stable matchings for market $(C_F,C_W)$.

\subsection{Quasi-corewise stability}\label{subseccion de quasi-core m-m}

In this subsection, we present a weakening of the corewise stability notion for many-to-many markets. Although this new notion allows for certain dominances, it requires that workers in the deviating coalition maintain the firms with which they are already employed, while potentially working for new firms. This notion extends to the many-to-many setting the worker-quasi-corewise notion presented in Subsection \ref{subseccion worker quasi en M-1}. 
We also study whether Theorem \ref{teorema characterization de quasi stable en m-1} holds in many-to-many settings.

Now, we formally define when a matching is in the worker-quasi-core.

\begin{definition}
    A matching $\mu$ is in the \textbf{worker-quasi-core} if there are a matching $\mu'$ and a coalition $S$ such that $\mu'$ dominates $\mu$ via $S$, it must satisfy that $\mu(w)\subseteq\mu'(w)$ for each $w\in S$.\footnote{Observe that this definition coincides with Definition \ref{defino worker-quasi-core}, since if a worker in coalition $S$ in the many-to-one market must continue maintaining the firms they are already matched with, and potentially incorporate new ones, the only option is that in matching $\mu$ they were previously unemployed.} 
\end{definition}
Let $\mathcal{C}^{\widetilde{\mathcal{QW}}}$ denote the worker-quasi-core. 
Unfortunately, a complete characterization like the one in the Theorem \ref{teorema characterization de quasi stable en m-1} cannot be obtained in the many-to-many setting. The following theorem states that the individually rational worker-quasi-core does not coincide with the worker-quasi-stable matching set; only one inclusion holds.

\begin{theorem}\label{quasi incluido en el cor IR}
 Let $(C_F,C_W)$ be a many-to-many market, then $\widetilde{\mathcal{QW}}\subseteq\widetilde{\mathcal{I}}\cap \mathcal{C}^{\widetilde{\mathcal{QW}}}.$
\end{theorem}
\begin{proof}
  Assume that $\mu\in \widetilde{\mathcal{I}} \setminus  \mathcal{C}^{\widetilde{\mathcal{QW}}}$ and we will prove that $\mu \notin \widetilde{\mathcal{QW}}.$ Then, there are a matching $\mu'$, a coalition $S$, and a worker $w\in S$ such that $\mu'$ dominates $\mu$ via $S$, and $\mu(w)\nsubseteq \mu'(w)$. We claim that  $\mu'(w)\neq \emptyset$. Assume otherwise, i.e. $\mu'(w)=\emptyset$. Since $\mu'$ dominates $\mu$ via $S$ and the individual rationality of $\mu$, $\emptyset =\mu'(w)=C_w(\mu(w)\cup\mu'(w))=C_w(\mu(w))=\mu(w).$ Then $\mu(w)=\emptyset$ contradicting that $\mu(w)\nsubseteq \mu'(w)$. Therefore  $\mu'(w)\neq \emptyset$, proving the claim. 
  
  Now, we claim that  $\mu'(w)\setminus \mu(w) \neq \emptyset$. Otherwise, since $\mu'(w)\neq \emptyset$, we have that  $\mu'(w)\subseteq \mu(w)$. Moreover, by the individual rationality of $\mu$, we have that $C_w(\mu(w)\cup\mu'(w))=C_w(\mu(w))=\mu(w).$ This contradicts that $\mu'(w) \succeq_w \mu(w)$ for $w\in S$. Thus, $\mu'(w)\setminus \mu(w) \neq \emptyset$, proving the claim. 

  \medskip
\noindent \textbf{Claim: $\boldsymbol{\mu'(w)\setminus \mu(w) \subseteq \widetilde{F}_w^{\mu}.}$}
 To see this, let $f'\in \mu'(w)\setminus \mu(w).$  
    Then, $w\in\mu'(f')\setminus \mu(f').$ Since  $f'\in S$ and $\mu'$ dominates $\mu$ via $S$,  $w\in \mu'(f')=C_{f'}(\mu'(f')\cup \mu(f')).$ By substitutability and the fact that  $w\in\mu'(f')\setminus \mu(f')$, we have that $w\in C_{f'}(\mu(f')\cup \{w\}).$ This implies that $f'\in \widetilde{F}_w^{\mu}.$ Hence, $\mu'(w)\setminus \mu(w) \subseteq \widetilde{F}_w^{\mu}$ proving  the claim. \medskip

  Since $\mu(w)\nsubseteq \mu'(w),$  there is $f\in \mu(w)\setminus\mu'(w).$ Note that $f\notin \mu'(w)$ and $\mu'(w) \succeq_w \mu(w)$ imply that $f\notin \mu'(w)=C_{w}(\mu(w)\cup \mu'(w)).$ By substitutability and the fact that $\mu'(w)\setminus \mu(w) \subseteq \widetilde{F}_w^{\mu}$, $f\notin C_{w}(\mu(w)\cup \widetilde{F}_w^\mu).$ This last fact together with $f\in \mu(w)\setminus\mu'(w)$ imply that $\mu(w)\nsubseteq C_{w}(\mu(w)\cup \widetilde{F}_w^\mu).$  Then, $\mu \notin \widetilde{\mathcal{QW}}.$ Therefore,   $\widetilde{\mathcal{QW}}\subseteq\widetilde{\mathcal{I}}\cap \mathcal{C}^{\widetilde{\mathcal{QW}}}$.
 \end{proof}

As a by-product of the previous theorem and the fact that $\widetilde{\mathcal{QW}}\neq \emptyset$, we have that $\mathcal{C}^{\widetilde{\mathcal{QW}}}\neq \emptyset.$

In what follows, we show that an individually rational matching that is in the worker-quasi-core, may not be worker-quasi-stable, i.e. in this case, the inclusion of Theorem \ref{quasi incluido en el cor IR} is strict. To do this, consider the preference profile  taken from Example 6.9 of \cite{roth1992two}: 
\begin{center}\noindent\begin{tabular}{l}
$ >_{f_1}: \{w_1w_2\}, \{w_2w_3\},\{w_1\},\{w_2\},\{w_3\},\emptyset  $\\
$ >_{f_2}: \{w_2w_3\}, \{w_1w_3\},\{w_2\},\{w_1\},\{w_3\},\emptyset$\\
$ >_{f_3}:\{w_1w_3\}, \{w_1w_2\},\{w_3\},\{w_1\},\{w_2\},\emptyset $\\
\end{tabular}~~~~~
\noindent\begin{tabular}{l}
$ >_{w_1}:\{f_1f_2\},\{f_2f_3\},\{f_1\},\{f_2\},\{f_3\} ,\emptyset$\\
$ >_{w_2}: \{f_2f_3\},\{f_1f_3\},\{f_2\},\{f_1\},\{f_3\},\emptyset $\\
$ >_{w_3}:\{f_1f_3\},\{f_1f_2\},\{f_3\},\{f_1\},\{f_2\},\emptyset $\\

\end{tabular}\medskip
\end{center} Note that each agent's preference can induce a choice function. For example, consider $>_{w_1}$, the choice function is induced as follows: $C_{w_1}(\{f_1,f_2,f_3\})=\{f_1,f_2\}$, $C_{w_1}(\{f_1,f_2\})=\{f_1,f_2\}$, $C_{w_1}(\{f_2,f_3\})=\{f_2,f_3\}$, $C_{w_1}(\{f_1,f_3\})=\{f_1\}$, $C_{w_1}(\{f_1\})=\{f_1\}$, $C_{w_1}(\{f_2\})=\{f_2\}$, $C_{w_1}(\{f_3\})=\{f_3\}$, $C_{w_1}(\emptyset)=\emptyset$.
Now, we consider the following example with the induced choice functions $C_F$ and $C_W.$

\begin{example}\label{ejemplo m-t-m cuasicore}
 Let $(C_F,C_W)$ be a many-to-many market, where $F = \{f_1,f_2,f_3\}$, $W = \{w_1,w_2,w_3\}$, and consider choices functions $C_F$ and $C_W$ induced by the preference profile presented above. Furthermore, consider the following matching:
$$
\mu=\begin{pmatrix}
w_1 & w_2 &w_3 \\
f_2f_3 &f_1f_3& f_1f_2  \\
\end{pmatrix}
.~$$
It can be checked that $\mu\in \widetilde{\mathcal{I}}$, and $\mu$ is undominated. Then $\mu\in \widetilde{\mathcal{C}}$ and , therefore, $\mu\in \mathcal{C}^{\widetilde{\mathcal{QW}}}.$ Note that, if we consider $w_1$ and take $K=\{f_1\}$, we have that $C_{w_1}(\mu(w_1)\cup K)=C_{w_1}(\{f_2,f_3\}\cup \{f_1\})=\{f_1,f_2\}$. Thus, $\mu(w_1)\nsubseteq C_{w_1}(\mu(w_1)\cup K)$, implying that $\mu \notin \widetilde{\mathcal{QW}}$ and, therefore, $\mu\notin \widetilde{S}.$
Therefore, $\widetilde{\mathcal{QW}} \subsetneq  \widetilde{\mathcal{I}}   \cap \mathcal{C}^{\widetilde{\mathcal{QW}}}  .$
\hfill $\Diamond$
\end{example}

\subsection{Quasi-setwise stability}\label{subseccion de setwise en m-m}
In this subsection, we present a weakening of the notion of setwise stability. We say that a matching is in the \textbf{setwise stable set} if it is individually rational and not dominated by any other matching via some coalition that forms new links between its members while preserving the links of agents outside the coalition. We also characterize the set of matching that satisfies our weaker notion of setwise stability with the worker-quasi-stable matching set.
Before presenting such a weaker notion of setwise stability, we need to introduce a special domination between matchings.
\begin{definition}\label{Defbloqueosetwise}
    Let $\mu,\mu'$ be two different matchings, and let $S\subseteq F\cup W$ be a non-empty coalition. We say that  \textbf{$\boldsymbol{\mu'}$ setwise dominates $\boldsymbol{\mu}$ via}  $\boldsymbol{S}$ if   $\mu'(S)\setminus\mu(S)\subseteq S$, $\mu'(a)\succeq_a\mu(a)$ for each $a\in S$, and there is $a'\in S$ such that $\mu'(a')\succ_{a'}\mu(a')$. for some $a'\in S.$ 
    \end{definition}
    According to this definition, a matching belongs to the setwise set if it is individually rational and not setwise dominated. We denote the set of setwise stable matchings by $\boldsymbol{\mathcal{SW}}$. Note that, based on the definitions of domination and setwise dominance (Definitions \ref{defino dominancia} and \ref{Defbloqueosetwise}, respectively), the set of setwise stable matchings is a subset of the core, i.e.  $\mathcal{SW} \subseteq \mathcal{\widetilde{C}}.$

    For weakening the notion of setwise stability, we allow some setwise dominations but impose that the workers in the coalition through which the setwise domination occurs do not cease to work for the firms they were already working for. Formally,
\begin{definition}\label{Defworker-quasi-setwise-stability}
   A matching $\mu$ is in the \textbf{ worker-quasi-setwise stable set} if it is individually rational and there are a matching  $\mu'$ and a
coalition $S$ such that $\mu'$ setwise dominates $\mu$ via $S$, it must satisfy that $\mu(w)\subseteq \mu'(w)$ for each $w \in S$.
\end{definition}

Let $\boldsymbol{ \mathcal{SW^{\widetilde{QW}}} }$ denote the worker-quasi-setwise stable set for market $(C_F,C_W)$. 
The following theorem characterizes the set of worker-quasi-stable matchings through the set of worker-quasi-setwise stable matchings. This statement is proved by double inclusion. 

\begin{theorem}
  Let $(C_F,C_W)$ be a many-to-many market, then   $\mathcal{SW^{\widetilde{QW}}}=\widetilde{\mathcal{QW}}.$
\end{theorem}
\begin{proof} We will prove the double inclusion in two cases.
\begin{itemize} 
    \item[$\boldsymbol{\subseteq )}$]  Assume that  $\mu\notin \widetilde{\mathcal{QW}}$, we will prove that $\mu\notin \mathcal{SW^{\widetilde{QW}}}.$ If $\mu \notin \widetilde{\mathcal{I}}$, then $\mu\notin \mathcal{SW^{\widetilde{QW}}}$, and we are done. So, take $\mu\in \widetilde{\mathcal{I}} \setminus \widetilde{\mathcal{QW}}$. Thus, there are $w\in W$ and $K\subseteq \widetilde{F}^\mu_w$ such that 
    \begin{equation}\label{noquasiworker}
        \mu(w)\nsubseteq C_w(\mu(w)\cup K).
    \end{equation}
    
     We consider $S=\{w,C_w(\mu(w)\cup K)\setminus \mu(w)\}.$ Let $\mu'$ be a matching such that $\mu'(w)=C_w(\mu(w)\cup K)$ and $\mu'(f)=C_f(\mu(f)\cup \{w\})$ for each $f\in S$.
      We claim that $\mu'$ setwise dominates $\mu$ via $S.$ To see this, we first show that $\mu'(S)\setminus \mu(S)\subseteq S.$  Note that, by definition of $\mu'$, $\mu'(w)\setminus \mu(w)=C_w(\mu(w)\cup K)\setminus \mu(w).$ So,  $\mu'(w)\setminus \mu(w)\subseteq S.$ Let $f\in S$. By definition of $\mu'$, $\mu'(f)\setminus \mu(f)=C_f(\mu(f)\cup \{w\})\setminus \mu(f).$  Since $f\in S,$ $f\in C_w(\mu(w) \cup K)\setminus \mu(w)$. Then, $f\in K.$ Note that, that $K\subseteq \widetilde{F}^\mu_w.$ By definition of $\widetilde{F}^\mu_w$, we have that $C_f(\mu(f)\cup \{w\})\setminus \mu(f)=\{w\}.$ Hence, $\mu'(f)\setminus \mu(f)=\{w\}.$ Therefore,   $\mu'(S)\setminus \mu(S)\subseteq S.$ 

Now, we show that $\mu'(a)\succeq_a \mu(a)$ for each $a\in S$. Let $w\in S$. By definition of $\mu'$, $C_w(\mu(w)\cup\mu'(w))=C_w(\mu(w)\cup C_f(\mu(w)\cup K))$. By Path Independence   $C_w(\mu(w)\cup C_w(\mu(w)\cup K))= C_w(\mu(w)\cup\mu(w)\cup K)=C_w(\mu(w)\cup K).$ By definition of $\mu'$ again,  $C_w(\mu(w)\cup K)=\mu'(w).$ Hence,  $C_w(\mu(w)\cup\mu'(w))=\mu'(w)$ implying that $\mu'(w)\succeq_w \mu(w)$. Let $f\in S$. Thus, $f\in C_w(\mu(w)\cup K)$ and $f\notin \mu(w).$ By definition of $\mu'$, $C_f(\mu(f)\cup\mu'(f))=C_f(\mu(f)\cup C_f(\mu(f)\cup\{w\})).$ By Path Independence and definition of $\mu'$, $C_f(\mu(f)\cup C_f(\mu(f)\cup\{w\}))= C_f(\mu(f)\cup\{w\})=\mu'(f).$ Hence,  $C_f(\mu(f)\cup\mu'(f))=\mu'(f)$ implying that $\mu'(f)\succeq_f \mu(f)$. Then, $\mu'$ setwise dominates $\mu$ via $S.$ Using $\eqref{noquasiworker}$ and definition of $\mu'$ we have that, $\mu(w)\nsubseteq \mu'(w).$ Therefore, $\mu\notin \mathcal{SW^{\widetilde{QW}}}.$
    
   \item[$\boldsymbol{\supseteq)}$] Assume that $\mu\notin \mathcal{SW^{\widetilde{QW}}}$ and we will prove that $\mu\notin \widetilde{\mathcal{QW}}.$  If $\mu\notin \mathcal{\widetilde{I}}$, then $\mu\notin \mathcal{\widetilde{QW}}$ and we are done. So, take $\mu\in \mathcal{\widetilde{I}}\setminus \mathcal{SW^{\widetilde{QW}}}.$
   Then, there are a matching $\mu'$, a coalition $S$, and a worker $w\in S$ such that $\mu'$ setwise dominates $\mu$ via $S$, and $\mu(w)\nsubseteq \mu'(w)$.  
   
   We claim that  $\mu'(w)\neq \emptyset$. Assume otherwise, i.e. $\mu'(w)=\emptyset$. Since $\mu'$ dominates $\mu$ via $S$ and the individual rationality of $\mu$, $\emptyset =\mu'(w)=C_w(\mu(w)\cup\mu'(w))=C_w(\mu(w))=\mu(w).$ Then $\mu(w)=\emptyset$ contradicting that $\mu(w)\nsubseteq \mu'(w)$. Therefore  $\mu'(w)\neq \emptyset$, proving the claim. 
  
  Now, we claim that  $\mu'(w)\setminus \mu(w) \neq \emptyset$. Otherwise, since $\mu'(w)\neq \emptyset$, we have that  $\mu'(w)\subseteq \mu(w)$. Moreover, by the individual rationality of $\mu$, we have that $C_w(\mu(w)\cup\mu'(w))=C_w(\mu(w))=\mu(w).$ This contradicts that $\mu'(w) \succeq_w \mu(w)$ for $w\in S$. Thus, $\mu'(w)\setminus \mu(w) \neq \emptyset$, proving the claim.
\medskip

\noindent \textbf{Claim: $\boldsymbol{\mu'(w)\setminus \mu(w) \subseteq \widetilde{F}_w^{\mu}.}$} 
 To see this, let $f'\in \mu'(w)\setminus \mu(w).$  
  Then, $w\in\mu'(f')\setminus \mu(f').$ Since  $f'\in S$ and $\mu'$ setwise dominates $\mu$ via $S$,  $w\in \mu'(f')=C_{f'}(\mu'(f')\cup \mu(f')).$ By substitutability and the fact that  $w\in\mu'(f')\setminus \mu(f')$, we have that $w\in C_{f'}(\mu(f')\cup \{w\}).$ This implies that $f'\in \widetilde{F}_w^{\mu}.$ Hence, $\mu'(w)\setminus \mu(w) \subseteq \widetilde{F}_w^{\mu}$ proving  the claim. \medskip
  
  Note that $\mu(w)\nsubseteq \mu'(w)$ implies that there is $f\in \mu(w)\setminus\mu'(w).$ Since $f\notin \mu'(w)$ and $\mu'(w) \succeq_w \mu(w)$, we have that $f\notin \mu'(w)=C_{w}(\mu(w)\cup \mu'(w)).$ By substitutability and the claim, $f\notin C_{w}(\mu(w)\cup \widetilde{F}_w^\mu).$ This fact together with $f\in \mu(w)\setminus\mu'(w)$ imply that $\mu(w)\nsubseteq C_{w}(\mu(w)\cup \widetilde{F}_w^\mu).$  Then, $\mu \notin \widetilde{\mathcal{QW}}.$ Therefore,   $\widetilde{\mathcal{QW}}\subseteq\widetilde{\mathcal{I}}\cap \mathcal{C}^{\widetilde{\mathcal{QW}}}$.
  \end{itemize} 
  
 By the double inclusion we have $\mathcal{SW^{\widetilde{QW}}}=\widetilde{\mathcal{QW}}.$
\end{proof}

As a by-product of the previous characterization and the fact that $\widetilde{\mathcal{QW}}\neq \emptyset$, we have that $\mathcal{SW^{\widetilde{QW}}}\neq \emptyset.$

The dual definition of firm-quasi-setwise stability can be obtained by requiring that firms in the deviating coalition do not fire any of the workers instead of imposing conditions on workers of the deviating coalition. This is, we now require that $\mu(f)\subseteq \mu'(f)$ for each $f \in S$ instead of $\mu(w)\subseteq \mu'(w)$ for each $w \in S$ in Definition \ref{Defworker-quasi-setwise-stability}. So, due to the symmetry of the market, the results of this section can be stated for firms.
Figure \ref{Diagrama many to many} depicts the relation between the notions of setwise stability, core, worker-quiasi-core, worker-quasi stability and  worker-quasi-setwise stability.

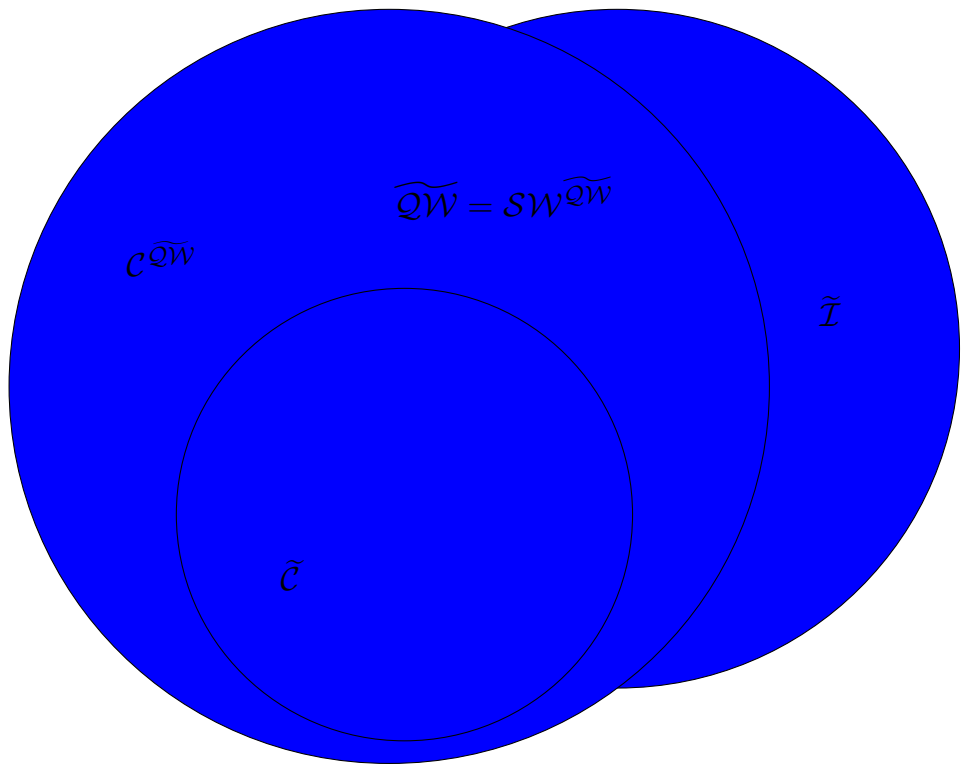
\begin{figure}[ht!]
    \centering
    \begin{tikzpicture}
    \fill[blue, opacity=0.1] (2, 0.5) circle (4.5cm);
    \draw (2, 0.5) circle (4.5cm);
    \node at (4.8, 1) {$\widetilde{\mathcal{I}}$};

     \fill[blue, opacity=0.15] (-1, 0) circle (5cm);
    \draw (-1, 0) circle (5cm);
    \node at (-4, 1.7) {$\mathcal{C^{\widetilde{QW}}}$};
    \node at (0.5, 2.5) {$\widetilde{\mathcal{QW}}=\mathcal{SW}^{\widetilde{{\mathcal{QW}}}}$};


      \fill[blue, opacity=0.2] (0, -1.2) circle (1.2cm);
    \draw (0, -1.2) circle (1.2cm);
    \node at (0, -1.2) {$\mathcal{SW} $};

   \fill[blue, opacity=0.2] (-0.8, -1.7) circle (3cm);
    \draw (-0.8, -1.7) circle (3cm);
    \node at (-2.3, -2.5) {$\widetilde{\mathcal{C}} $};
\end{tikzpicture}
    \caption{Relations in many-to-many markets}
    \label{Diagrama many to many}
\end{figure}

\subsection{Pairwise stability, Corewise stability and Quasi-stability}\label{subseccion de core y quasi estabilidad en m-m}
In this final subsection, we study the relationship between the worker-quasi-stable set, the firm-quasi-stable set, the pairwise stable set, and the core. Note that, unlike the many-to-one model, where it is known that the core coincides with the set of pairwise stable matchings, in the many-to-many model there may be unstable matchings in the core, as observed in Example \ref{ejemplo m-t-m cuasicore}. Additionally, the core matching $\mu$ is not worker-quasi-stable and, therefore, neither pairwise stable. Moreover, it is easily observed that $\mu$ is nether firm-quasi-stable. The following example presents a market where a core matching can be worker-quasi-stable but not firm-quasi-stable (thus not pairwise stable). Consider a modification of the preferences from Example 6.9, taken from \cite{roth1992two}:

\begin{center}
\begin{tabular}{l}
$ >_{f_1}: \{w_1w_2\}, \{w_2w_3\},\{w_1\},\{w_2\},\{w_3\},\emptyset  $\\
$ >_{f_2}: \{w_2w_3\}, \{w_1w_3\},\{w_2\},\{w_1\},\{w_3\},\emptyset$\\
$ >_{f_3}:\{w_1w_2w_3\},\{w_1w_3\},\{w_1w_2\},\{w_2w_3\},\{w_3\},\{w_1\},\{w_2\},\emptyset $\\
\\
$ >_{w_1}:\{f_1f_2f_3\},\{f_1f_2\},\{f_2f_3\},\{f_1f_3\},\{f_1\},\{f_2\},\{f_3\},\emptyset $\\
$ >_{w_2}: \{f_1f_2f_3\},\{f_1f_3\},\{f_1f_2\},\{f_2f_3\},\{f_2\},\{f_1\},\{f_3\},\emptyset $\\
$ >_{w_3}:\{f_1f_2f_3\},\{f_1f_3\},\{f_1f_2\},\{f_2f_3\},\{f_3\},\{f_1\},\{f_2\} ,\emptyset$\\
\end{tabular}
\end{center}\medskip
Note that, as in Example \ref{ejemplo m-t-m cuasicore}, each agent induces a choice function based on their preferences. Thus, 
 the market $(C_F,C_W)$ of the following example is well defined.
\begin{example}
 Let $(C_F,C_W)$ be a market with $F = \{f_1,f_2,f_3\}$ and $W = \{w_1,w_2,w_3\}$.
 Consider the following matching:
$$
\mu=\begin{pmatrix}
w_1 & w_2 &w_3 \\
f_2f_3 &f_1f_3& f_1f_2f_3  \\
\end{pmatrix}
.~$$
It can be checked that $\mu \in \widetilde{\mathcal{I}}$, and $\mu$ is undominated. Therefore, $\mu \in \widetilde{\mathcal{C}}$. Note that $\mu \notin \widetilde{\mathcal{S}}$, as $(f_1, w_1)$ is a blocking pair. Moreover, since $(f_1, w_1)$ is the only blocking pair, $\widetilde{F}^{\mu}_{w_1} = \{f_1\}$ and $C_{w_1}(\mu(w_1) \cup \widetilde{F}^{\mu}_{w_1}) = C_{w_1}(\{f_2, f_3\} \cup \{f_1\}) = \{f_1, f_2, f_3\}$. Thus, $\mu(w_1) \subseteq C_{w_1}(\mu(w_1) \cup \widetilde{F}^{\mu}_{w_1})$, implying that $\mu \in \widetilde{\mathcal{QW}}$. 
\hfill $\Diamond$
\end{example}
Due to the symmetry of the market, this same example, by exchanging the roles of firms and workers, shows that there is a core matching that is firm-quasi-stable but not worker-quasi-stable, and therefore not pairwise stable.

In the following theorem, we find a relation between all of these notions. To do this, we now restrict ourselves to matchings that are both worker-quasi-stable and firm-quasi-stable, but not pairwise stable, and analyze their relationship with the core. Next, we show that these matchings are always dominated, i.e. they are not in the core.

\begin{theorem}\label{theorem core and quasi-stability}
Let $(C_F,C_W)$ be a many-to-many market, if a matching is both worker-quasi-stable and firm-quasi-stable, then the matching is either pairwise stable or is dominated by another matching via some coalition, i.e. $ \widetilde{\mathcal{QW}}\cap \widetilde{\mathcal{QF}}\subseteq \widetilde{\mathcal{S}}\cup\widetilde{\mathcal{C}}^c.$\footnote{$\widetilde{\mathcal{C}}^c$ denotes the complement of the core, i.e. the set of dominated matchings.}
\end{theorem}
\begin{proof}
    Assume that $\mu\in \widetilde{\mathcal{QW}}\cap \widetilde{\mathcal{QF}}\setminus\widetilde{\mathcal{S}} $ and we will prove that $\mu \in \widetilde{\mathcal{C}}^c.$ Then, there is a blocking pair $(f,w)$ of $\mu$. Thus, $f\in C_w(\mu(w)\cup\{f\})$ and $w\in C_f(\mu(f)\cup\{w\})$. 
    Let $$\widetilde{W}_f^\mu=\{w \in W : f \in C_w(\mu(w)\cup\{f\})\}.$$
    Thus, by definition of $\widetilde{F}^{\mu}_w$ and $\widetilde{W}^{\mu}_f$, we have 
    \begin{equation}\label{ecu 1 pureba theorem core and quasi-stability}
        \{f\}\subseteq \widetilde{F}^{\mu}_w\text{ and }\{w\}\subseteq \widetilde{W}^{\mu}_f.
    \end{equation}
    Recall that by definition of $\widetilde{\mathcal{QW}}$ and $\widetilde{\mathcal{QF}}$, we have that 
    \begin{equation}\label{ecu 2 pureba theorem core and quasi-stability}
        \mu(w) \subseteq C_w(\mu(w)\cup \{f\})\text{ and }\mu(f) \subseteq C_f(\mu(f)\cup \{w\}).
    \end{equation}  
    Since, by hypothesis $\mu\in \widetilde{\mathcal{QW}}\cap \widetilde{\mathcal{QF}}$, and the fact that $(f,w)$ is a blocking pair of $\mu$, together with \eqref{ecu 1 pureba theorem core and quasi-stability} and \eqref{ecu 2 pureba theorem core and quasi-stability} imply that  
    \begin{equation}\label{ecu 3 pureba theorem core and quasi-stability}
        \mu(w)\cup \{f\}= C_w(\mu(w)\cup \{f\})\text{ and }\mu(f)\cup \{w\}=C_f(\mu(f)\cup \{w\}).
    \end{equation} 
    Define a matching $\mu'$ such that $$\mu'(a)=\begin{cases}
        \mu(a)\cup\{f\}& \text{if } a=w,\\
        \mu(a)\cup\{w\}& \text{if } a=f,\\
        \mu(a)& \text{otherwise}.\\
            \end{cases}$$
 Note that the definition of $\mu'$ and \eqref{ecu 3 pureba theorem core and quasi-stability} imply that $\mu'$ is individually rational. Now, we claim that $\mu'$ dominates $\mu$ via the coalition $F\cup W$, and thus $\mu\notin \widetilde{\mathcal{C}}.$
  First, by definition of $\mu'$, $\mu'(F\cup W)\subseteq F\cup W.$ Second, to prove $\mu'(a)\succeq_a \mu(a)$ for each $a\in F\cup W.$ Consider $a\notin \{f,w\}.$ By definition of $\mu'$ and the individual rationality of $\mu'$, we have $C_a(\mu'(a)\cup \mu(a))=C_a(\mu'(a)\cup \mu'(a))=C_a(\mu'(a))=\mu'(a).$ Then, $\mu'(a)\succeq_a \mu(a)$ for $a\notin \{f,w\}.$
  Now, consider the firm $f$ of the blocking pair.  By definition of $\mu'$ and \eqref{ecu 3 pureba theorem core and quasi-stability}, $C_f(\mu'(f)\cup \mu(f))=C_f(\mu(f)\cup\{w\} \cup \mu(f))=C_f(\mu(f)\cup\{w\})=\mu(f)\cup\{w\}=\mu'(f).$ Then, $\mu'(f)\succeq_f \mu(f).$ By the same reasoning, we can show that $\mu'(w)\succeq_w \mu(w)$ for $w$ of the blocking pair. Hence, $\mu'(a)\succeq_a \mu(a)$ for each $a\in F\cup W$. Then, $\mu'$ dominates $\mu$ via the coalition $F\cup W$, and the claim is proven. Therefore, $\mu \in \widetilde{\mathcal{S}}\cup\widetilde{\mathcal{C}}^c.$  
  \end{proof}
  
Figure \ref{Diagrama many to many bis} depicts the relation between the notions of worker-quasi-stability, firm-quasi-stability, pairwise stability, and corewise stability.

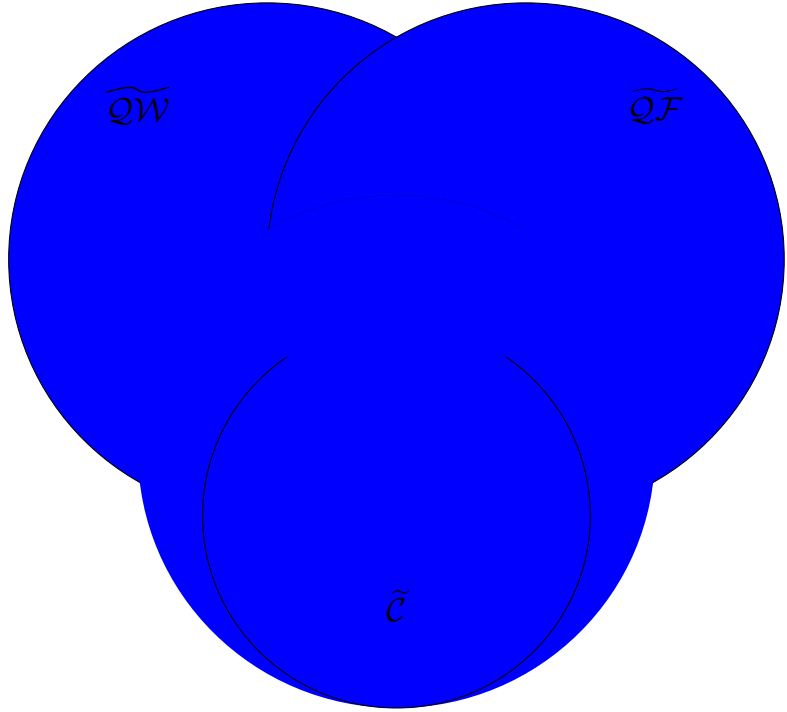
\begin{figure}[ht!]
    \centering
\begin{tikzpicture}[scale =1.7]
    \fill[blue, opacity=0.2] (0, 0) circle (2cm);
    \draw (0, 0) circle (2cm);
     \node at (-1, 1.2) {$\widetilde{\mathcal{QW}}$};

    \fill[blue, opacity=0.2] (2, 0) circle (2cm);
    \draw (2, 0) circle (2cm);
     \node at (3, 1.2) {$\widetilde{\mathcal{QF}}$};

    \fill[blue, opacity=0] (1, -1.5) circle (2cm); 
\node at (1,0) {$\widetilde{\mathcal{S}}  $};

    \fill[blue, opacity=0.2] (1, -2) circle (1.5cm);
    \draw (1, -2) circle (1.5cm);
\node at (1, -2.7) {$\widetilde{\mathcal{C}} $};
    \begin{scope}
        \clip (0, 0) circle (2cm); 
        \clip (2, 0) circle (2cm); 
        \clip (1, -1.5) circle (2cm); 
        \draw  (1, -1.5) circle (2cm);
        \fill[blue, opacity=0.2] (-3, -3) rectangle (5, 3); 
    \end{scope}

\end{tikzpicture}
 \caption{Relations in many-to-many markets}
    \label{Diagrama many to many bis}
\end{figure}

\section{Concluding Remarks}\label{Concluding Remarks}

The main contribution of this paper is to examine weakenings of the concepts of corewise stability and setwise stability in matching markets, following the idea of weakening pairwise stability as known in (pairwise) quasi-stability. The concept of quasi-stability has been extensively studied in the literature; to the best of our knowledge, this is the first paper to investigate weakenings of corewise and setwise stability.

For a many-to-one market, we show that the set of individually rational matchings in the worker-quasi-core is characterized by the set of worker-quasi-stable matchings. Despite the market’s asymmetry, we also find that the set of individually rational matchings in the firm-quasi-core corresponds to the set of firm-quasi-stable matchings.

For a many-to-many market, we demonstrate that each worker-quasi-stable matching is individually rational and belongs to the worker-quasi-core. We also show that the set of worker-quasi-setwise stable matchings corresponds to the set of worker-quasi-stable matchings As a final result, we show that a matching that is both worker-quasi-stable and firm-quasi-stable is either pairwise stable or dominated by another matching through certain deviating coalitions.

In this article, we analyze a many-to-many market and establish that every worker-quasi-stable matching is individually rational and resides within the worker-quasi-core. We further demonstrate that the set of worker-quasi-setwise stable matchings aligns with the set of worker-quasi-stable matchings. Lastly, we prove that any matching satisfying both worker-quasi-stability and firm-quasi-stability is either pairwise stable or is dominated by an alternative matching through specific coalition deviations.

Many papers that examine the notion of quasi-stability find that this set has a lattice structure \citep[see][among others]{wu2018lattice,bonifacio2022lattice,bonifacio2024envyfreelattice,bonifacio2024lattice,yang2024existence}. Based on our characterizations in many-to-one markets, we conclude that both the set of individually rational matchings within the worker-quasi-core and the set of individually rational matchings within the firm-quasi-core exhibit lattice structure. Furthermore, in many-to-many markets, the worker-quasi-setwise stable set also possesses a lattice structure. Thus, due to the symmetry of the many-to-many market, the firm-quasi-setwise stable set also possesses a lattice structure.

Consider a many-to-one market where firms have responsive preferences, a preference structure that is more restrictive than substitutable preferences, and consider allowing agents to have indifferences in their preferences. In this setting, \cite{bonifacio2024core} investigate three notions of the core (core, strong core, and super core) and their relationships with the known (pairwise) stability notions (stability, strong stability, and super stability) introduced in \cite{irving1994stable}. A potential line of inquiry could involve adapting the (pairwise) quasi-stability notions for the three (pairwise) stability notions introduced in \cite{irving1994stable} and exploring their relationships with the quasi-corewise stability concepts adapted to the three core notions presented in \cite{bonifacio2024core}.

\end{document}